\documentclass[12pt]{article}

\usepackage{pgf}

\usepackage[english]{babel}
\usepackage[utf8]{inputenc}

\usepackage{amsmath,amssymb}
\usepackage{geometry}
\usepackage{array}
\usepackage{xspace}

\usepackage{amsmath,amsfonts,amssymb,amsopn,amsthm}
\theoremstyle{plain}

\newtheorem{lemma}{Lemma}

\newtheorem{proposition}{Proposition}

\newtheorem{theorem}{Theorem}
\theoremstyle{definition}

\theoremstyle{remark}
\newtheorem{remark}{Remark}

\newcommand{\C}{{\mathbb C}}

\newcommand{\F}{{\mathbb F}}

\newcommand{\Z}{{\mathbb Z}}
 
 \newcommand{\sN}{{\mathcal N}}
\newcommand{\sL}{{\mathcal L}}

\newcommand{\be}{\begin{eqnarray}}
\newcommand{\ee}{\end{eqnarray}}
\newcommand{\nn}{{\nonumber}}

\newcommand{\Tr}{{\rm Tr}}

\makeatother

\newcommand{\Wa}[1]{\widehat{\chi_{#1}}}
\newcommand{\Supp}{{\rm Supp}}

\begin{document} 

\title{Strongly regular graphs from weakly regular plateaued functions}

\author{Sihem Mesnager{\thanks{Department of Mathematics, University of Paris VIII, University of Paris XIII, CNRS, UMR 7539 LAGA and Telecom ParisTech, Paris, France.
Email: smesnager@univ-paris8.fr}} \and Ahmet S{\i}nak{\thanks{ 
Department of Mathematics and Computer Sciences at Necmettin Erbakan University, Turkey and LAGA, UMR 7539, CNRS at Universities
of Paris VIII and Paris XIII, France.
 Email: sinakahmet@gmail.com}}}

\date{}

\maketitle

\begin{abstract}

The paper provides the first constructions of strongly regular graphs  and association
schemes from weakly regular plateaued functions over finite fields of odd characteristic.
We generalize the construction method of strongly regular graphs from weakly regular bent functions given by Chee et al. in [Journal of Algebraic Combinatorics, 34(2), 251-266, 2011] to weakly regular plateaued functions. In this framework, we   construct strongly regular graphs with three types of parameters  from weakly regular plateaued functions with some homogeneous conditions. We also construct a family of association schemes of class $p$ from weakly regular $p$-ary plateaued functions.

 \end{abstract}
{\bf Keywords} Association schemes,  Partial difference sets,  Strongly regular graphs,  Weakly regular plateaued functions.

\section{Introduction}
Certain combinatorial objects such as  association schemes and strongly regular graphs have a diverse applications in many areas, especially coding theory and cryptography. 
It was shown in a few paper \cite{chee2011strongly,feng2010partial,pott2011association,tan2010strongly} that   these combinatorial objects 
may be  constructed from  weakly regular bent functions over finite fields of  odd characteristic. With in this framework,  we make use of in this paper for the first time the weakly regular plateaued functions to construct association schemes and  strongly regular graphs.
The paper is organized as follows. The rest of this section settles the necessary notations and  background of this paper. In Section \ref{SectionConstruction}, we first  give a family of   $p$-class association schemes  from weakly regular  plateaued functions over finite fields of   characteristic $p$.  We next
construct  three types of strongly regular graphs  from weakly  regular plateaued functions over finite fields of odd characteristic. 

\subsection{Difference Sets}
Group rings and character theory are necessary tools to handle difference sets. The reader is referred  to \cite{passman2011algebraic} for the group rings,  to \cite{lidl1997finite} for the character theory  on finite fields and to  \cite{ma1994survey} for  difference sets.

Let $G$ be a (multiplicative) group of order $v$. A $k$-subset $D$ of $G$ is called a $(v,k,\lambda)$ \textit{difference set} if each non-identity element of $G$ can be represented as $gh^{-1}$ $(g,h\in D, g\neq h)$ in exactly $\lambda$ ways.    Then,  $D$ is a $(v,k,\lambda)$  difference set in $G$ if and only if the following equation holds in the group ring $\Z[G]$:
\be\nn
 DD^{(-1)}=(k-\lambda)1_G+ \lambda G,
\ee
where $D=\sum_{d\in D}d$, $D^{(-1)}=\sum_{d\in D}d^{-1}$  and $G=\sum_{g\in G}g$. Moreover, 
A $k$-subset $D$ of $G$ is  called a $(v,k,\lambda,\mu)$ \textit{partial difference set (PDS)} if each non-identity element in $D$ (resp., $G\setminus D$) can be represented as $gh^{-1}$ $(g,h\in D, g\neq h)$ in exactly $\lambda$ (resp., $\mu$) ways.
 It is  usually assumed that the identity element $1_G$ of $G$ is not involved in $D$.   By the group ring notation, 
 $D$   is a $(v,k,\lambda,\mu)$ PDS in $G$ if and only if 
\be\label{Condition}
DD^{(-1)}=(k-\mu)1_G+(\lambda-\mu)D+\mu G.
\ee
For  $\lambda\ne \mu$, we have  $D^{(-1)}=D$. A well-known example of PDS is the Paley PDS. \\
Strongly regular graphs (SRGs) are the certain combinatorial  objects associated with partial difference sets. Given a PDS $D$ with $\lambda\neq \mu$  in $G$, one can construct a strongly regular Cayley graph, $Cay(G,D)$, whose vertex set is $G$  and two vertices $g,h$ are joined by an edge if and only if $gh^{-1}\in D$. 
The following proposition states the connection between SRGs and PDSs.

\begin{proposition}\cite{ma1994survey}
A $k$-subset $D$ of $G$ is a $(v,k,\lambda,\mu)$-PDS in $G$ with $1_G\notin D$ and $D^{(-1)}=D$ if and only if a Cayley graph $Cay(G,D)$  is a $(v,k,\lambda,\mu)$ SRG.
\end{proposition}

\subsection{Association schemes and Schur rings}

  Let $V$ be a finite set of vertices and $\{G_0,G_1, \ldots, G_d\}$ be binary relations on $V$ with  $G_0=\{(x,x) : x\in V\}$. Then the decomposition $(V; G_0,G_1,\ldots,G_d)$ (shortly, it may be represented as  $(V,\{G_i\}_{0\leq i\leq d})$) is called \textit{an association scheme} of class $d$ on $V$ provided that  the following properties hold:
\begin{itemize}
\item [-] $V\times V=G_0 \cup G_1 \cup \cdots \cup G_d$ and $G_i\cap G_j=\emptyset$ for $i\neq j$;
\item [-] $^t G_i=G_{i'}$ for some $i'\in\{0,1,\ldots,d\}$, where $^t G_i =\{(x,y) : (y,x)\in G_i\}$; (If $i'=i$, then we call $G_i$ symmetric.)
\item [-] for $i,j,k\in\{0,1,\ldots,d\}$ and $x,y\in V$ with $(x,y)\in G_k$,  the number $p_{ij}^k:=\#\{z\in V : (x,z)\in G_i, (z,y)\in G_j\}$ is  a constant.
\end{itemize}

An association scheme is said to be \textit{symmetric} if each $G_i$ is symmetric. 
 One of the well-known construction methods of association schemes is to use Schur rings. Let $G$ be a finite Abelian group  and $D_i$, $0\leq i\leq d$, be nonempty subsets of $G$ with the following properties:
\begin{itemize}
\item [-] $D_0=\{1_G\}$;
\item [-] $G=D_0 \cup D_1 \cup \cdots \cup D_d$ and $D_i\cap D_j=\emptyset$ for $i\neq j$;
\item [-] $D^{(-1)}_i=D_{i'}$ for some $i'\in\{0,1,\ldots,d\}$, where $D^{(-1)}_i =\{g^{-1} : g\in D_i\}$;
\item [-]  $D_i D_j=\sum_{k=0}^{d}p_{ij}^k D_k$ for all $0\leq i,j\leq d$, where $p_{ij}^k$ are integers.
\end{itemize}
Then the subset $\langle D_0,\ldots, D_d \rangle$ in $\C[G]$ spanned by $D_0,\ldots, D_d$  is called \textit{Schur ring} over $G$. The configuration  $(G,\{G_i\}_{0\leq i\leq d})$  forms an association scheme of class $d$ on $G$, where $G_i:=\{ (g,h) : gh^{-1}\in D_i\}$ for $0\leq i\leq d$.
For more detail on association schemes,  the reader is referred to  \cite{van2010some}.

\subsection{Weakly regular plateaued function}
In this subsection, we  set the basic notations and previous results related to weakly regular plateaued functions. 
We first fix the following notations throughout  this paper unless otherwise stated.
\begin{itemize}
\item [-]  $\Z$  is the  ring of integers and $\mathbb{C}$ is the field of complex numbers,
\item [-]  $p$ is an odd prime and $q=p^n$ is an $n$-th  power of $p$ with  a positive integer $n$,
\item   [-] $\F_{p^n}$ is the finite field with $p^n$ elements and $\F_{p^n}^{\star}:=\F_{p^n}\setminus \{ 0\} $ is a cyclic group, 
\item   [-]
The  trace of $\alpha\in\F_{p^n}$ over $\mathbb {F}_{p}$ is defined by $\Tr^n(\alpha)=\alpha+\alpha^{p}+\alpha^{p^{2}}+\cdots+\alpha^{p^{n-1}}$,
\item  [-] $\xi_p=e^{2\pi i/p}$ is the complex primitive $p$-th root of unity, where $i=\sqrt {-1}$ is the complex primitive $4$-th root of unity,
\item [-] $SQ$ and $NSQ$  denote respectively  the set of all  squares and non-squares in  $\F_p^{\star}$,
\item [-]  $p^*=\eta_0(-1)p$, where $\eta_0$  is the quadratic character of $\F_p^{\star}$. 

\end{itemize}

The function $\chi$ from $\F_q$ to $\C$ defined as 
$\chi(x)=\xi_p^{\Tr_p^q(x)}$
 for    $x\in \F_q$ is called \textit{the canonical additive character} of $\F_q$.
An additive character of $\F_q$ has additive property: $\chi(x_1+x_2)=\chi(x_1)\chi(x_2)$  for all $x_1,x_2\in\F_q$.
Let $f :  \F_{p^n} \longrightarrow  \mathbb {F}_{p}$  be a $p$-ary function.
The Walsh transform of $f$ is given by:
$$\Wa {f}(\beta )=\sum_{x\in  \F_{p^n}} {\xi_p}^{{f(x)}-\Tr^n (\beta x)}, \;\; \beta\in  \F_{p^n}.$$
 A function  $f$ is said to be \textit{balanced} over $\F_p$ if 
 $\Wa f(0)=0$; otherwise, $f$ is called \textit{unbalanced}.
The notion of  plateaued functions was first introduced    in characteristic 2  by Zheng and Zhang   \cite{zheng1999plateaued} in 1999.   A function $f$ is  said to be  $p$-ary $s$-plateaued   if $|\widehat{\chi_f}(\beta)|^2\in\{0,p^{n+s}\}$ for every $\beta\in \F_{p^n}$, where $s$ is an integer with $0\leq s\leq n$. In the case of $s=0$, the $0$-plateaued function is the bent function.  From \cite{mesnager2017WCC}, an  $s$-plateaued $f$  is said to be  \emph{weakly regular}  if  there exists a constant complex number $u$ having unit magnitude 
such that 
$$
\Wa {f}(\beta)\in  \{ 0, up^{(n+s)/2}\xi_p^{g(\beta)} \}
$$
 for  all $\beta\in \F_{p^n}$, where  $g$ is  a $p$-ary function  over $\F_{p^n}$ with $g(\beta)=0$ for all $\beta\in  \F_{p^n} \setminus  \Supp(\Wa f)$; otherwise,  $f$ is said to be \emph{non-weakly regular}. Indeed, weakly regular    $f$ is said to be   \emph{regular} if $u=1$. For example, all quadratic functions are the weakly regular plateaued functions (see \cite[Proposition 3]{IEEE}).
The Walsh support of  plateaued   $f$ is defined by  $\Supp(\widehat{\chi_f})=\{\beta\in  \F_{p^n}:  |\widehat{\chi_f}(\beta)|^2= p^{n+s}\}$.  
 The \textit{Parseval identity} is given by $\sum_{\beta \in \F_{p^n}} | \widehat{\chi_f}(\beta )|^{2}=p^{2n}.$ The absolute Walsh distribution of   plateaued functions  follows directly from the Parseval identity. 
\begin{lemma} \label{SupportLemma} 
Let $f:\F_{p^n}\to\F_{p}$ be an $s$-plateaued function. Then for $\beta \in\F_{p^n}$, $| \widehat{\chi_f}(\beta)|^2$ takes $p^{n-s}$ times the value $p^{n+s}$ and $p^n-p^{n-s}$ times the value $0$.
\end{lemma}

\begin{lemma}\label{WalshFact} \cite{mesnager2017WCC}
Let  $f:\F_{p^n}\to\F_{p}$ be a weakly regular  $s$-plateaued function. Then
for  all $\beta\in \Supp(\Wa f)$,  we have
$
 \Wa {f}(\beta)=\epsilon \sqrt{p^*}^{n+s} \xi_p^{g(\beta)},
$
 where $\epsilon=\pm 1$ is    the sign of  $ \Wa f$ and  $g$ is a $p$-ary function over  $\Supp(\Wa f)$. 
\end{lemma} 

 Let $f:\F_{p^n}\to\F_{p}$ be a weakly regular  $p$-ary $s$-plateaued unbalanced function, where $0\leq s\leq n$. In \cite{IEEE}, we denote by  $WRP$   the set of these functions satisfying the following two properties: $f(0)=0$ and there  exists a positive even integer $h$ with $\gcd(h-1,p-1)=1$ such that 
$f(ax)=a^hf(x)$  for any $a\in\F_p^{\star}$ and $x\in \F_{p^n}$.
In the following section, we make use of  these functions to construct   association schemes and  partial difference sets.  

\begin{lemma}\cite{IEEE}\label{Walshsupport}
Let $f\in WRP$. Then for any $\beta\in\Supp(\Wa f)$ (resp., $\beta\in\F_{p^n}\setminus\Supp(\Wa f)$), we have
$z\beta\in\Supp(\Wa f)$ (resp., $z\beta\in\F_{p^n}\setminus\Supp(\Wa f)$) for every $z\in\F_p^{\star}$.
\end{lemma}

\begin{proposition}\cite{IEEE}\label{Proposition5}
Let $f\in WRP$, then there exists a  positive even integer $l$ with $\gcd(l-1,p-1)=1$ such that  $g(a \beta)=a^lg(\beta)$  for any $a\in\F_p^{\star}$ and $\beta \in \Supp(\Wa f)$.
\end{proposition}
We now define the subsets
\be\label{Subsets}
D_{f,j}=\{ x\in\F_{p^n}: f(x)=j\}
\ee
 for every $j\in\F_p$ and 
\be\label{Elements}
 \sL_t=\sum_{j=0}^{p-1}D_{f,j}\xi_p^{jt}
\ee
 for every $t\in\F_p$, in particular, we have $\sL_0=\F_{p^n}$ (here $\sL_t$ can be seen as an element of the group ring $\C[(\F_{p^n},+)]$). 
Clearly, we have 
$
\sum_{j=0}^{p-1}D_{f,j}=\F_{p^n} \mbox{ and } D_{f,i}\cap D_{f,j}=\emptyset$ for $i\neq j$,
 which implies that  $D_{f,j}$'s are the partition of $\F_{p^n}$. Moreover, for every $j\in\F_p$ we have  $D_{f,j}^{(-1)}=D_{f,j}$  and 
the following  may be  easily observed 
\be\label{DD}
D_{f,j}=\frac{1}{p} \sum_{t=0}^{p-1}\sL_t\xi_p^{-jt}.
\ee

\begin{lemma}\cite{IEEE}\label{Lemma7}
Let $n+s$ be an even integer and $f:\F_{p^n}\to\F_p$ be an unbalanced function with $\Wa f (0)=\epsilon \sqrt{p^*}^{n+s}$, where $\epsilon=\pm 1$  is    the sign of  $ \Wa f$. We  denote by $\sN_f(j)$ the size of the set $D_{f,j}$ defined in (\ref{Subsets}) for every $j\in\F_p$.
  Then,  
\be\nn
\sN_{f}(j)=\left\{\begin{array}{ll}
p^{n-1}+\epsilon \eta_0 (-1)(p-1)\sqrt{p^*}^{n+s-2},  & \mbox{ if } j=0, \\
p^{n-1}-\epsilon \eta_0 (-1)\sqrt{p^*}^{n+s-2},& \mbox{ if } j\in\F_p^{\star}.
 \end{array}\right.
\ee
\end{lemma}

\section{The construction of association schemes and strongly regular graphs}\label{SectionConstruction}
This section generalizes the construction methods given in  \cite{chee2011strongly,feng2010partial,pott2011association} of association schemes and partial difference sets from weakly regular bent functions  to weakly regular plateaued functions over finite fields of odd characteristic. Throughout this  section, the group $G$ is an additive group of $(\F_{p^n},+)$.

\subsection{A family of association schemes of class $p$  from weakly regular $p$-ary plateaued functions}
In this subsection, we give a family of   $p$-class association schemes  from weakly regular  plateaued functions over  $\F_p$.  To do this, we first need the following properties of the $\sL_t$'s. 

\begin{lemma} \label{LemmaAsso}
Let $f\in WRP$ and $l$ be given by Proposition \ref{Proposition5}. Let $D_{f,j}$  and $\sL_j$ be given by (\ref{Subsets}) and (\ref{Elements}) for every $j\in\F_p$, respectively.
For  $k,t\in\F_p^{\star}$, 
\begin{itemize}
\item [$i.)$] if  $k+t\neq 0$,   we have $\sL_t\sL_k=\epsilon\eta^{n+s}_0(ktv)\sqrt{p^*}^{n+s}\sL_v$, where  $\epsilon=\pm 1$ is    the sign of  $ \Wa f$ and
$v=(k^{1-l}+t^{1-l})^\frac{1}{{1-l}}$; 
\item [$ii.)$] $\sL_t\sL_{-t}=p^{n}$;
\item [$iii.)$] $\sum_{t=1}^{p-1}\sL_t\sL_0\xi_p^{-jt}=(p\#D_{f,j}-p^n)\F_{p^n}$ for $j\in\F_p$. 
\end{itemize}
\end{lemma}
\begin{proof}
For $\beta\in\F_{p^n}$,
applying an additive character $\chi_\beta$ of $\F_{p^n}$ to $\sL_t$   for $t\in\F_p^{\star}$, we have 
\be\nn
 \begin{array}{ll}
\chi_\beta(\sL_t)&=\displaystyle\sum_{j=0}^{p-1}\chi_\beta(D_{f,j})\xi_p^{jt}=
\displaystyle\sum_{x\in \F_{p^n}} \xi_p^{tf(x)+\Tr(\beta x)}=
\sigma_t(\Wa {f}(t^{-1}\beta))\\
&=
\left\{ \begin{array}{ll}
0,& \mbox{ if }  \beta\in\F_{p^n}\setminus\Supp(\Wa f),\\
\epsilon\eta_0^{n+s}(t) \sqrt{p^*}^{n+s} \xi_p^{t^{1-l}g(\beta)},& \mbox{ if }  \beta\in\Supp(\Wa f),
\end{array}\right.
\end{array}
\ee
 where we used Lemma \ref{Walshsupport} and Proposition \ref{Proposition5} in the last equality and  $g$ is a $p$-ary function over  $\Supp(\Wa f)$. We now prove  $(i)$.  If $k+t\neq 0$, then for $\beta\in\Supp(\Wa f)$
\be\nn
\begin{array}{ll}
\chi_\beta(\sL_t\sL_k)&= 
 \epsilon^2\eta_0^{n+s}(tk) (\sqrt{p^*}^{n+s})^2 \xi_p^{(k^{1-l}+t^{1-l})g(\beta)} \\
&=
\epsilon \eta_0^{n+s}(tkv^{-1}) \sqrt{p^*}^{n+s}
\epsilon \eta_0^{n+s}(v) \sqrt{p^*}^{n+s} \xi_p^{v^{1-l}g(\beta)}\\
&=\epsilon \eta_0^{n+s}(tkv) \sqrt{p^*}^{n+s}\chi_\beta(\sL_v),
\end{array}
\ee
where we used the fact that $\eta_0(tkv^{-1})=\eta_0(tkv)$ for $v\in\F_p^{\star}$ in the last equality. This implies that we have $\sL_t\sL_k=\epsilon \eta_0^{n+s}(tkv) \sqrt{p^*}^{n+s}\sL_v$ with $k^{1-l}+t^{1-l}=v^{1-l}.$ Notice that $a^{1-l}+b^{1-l}=0$ if and only if $a+b=0$ for $a,b\in\F_p$  since $(l-1,p-1)=1$. We now prove $(ii)$.  For  $t\in\F_p^{\star}$, we clearly have 
\be\nn
\chi_\beta(\sL_t\sL_{-t})=
\left\{ \begin{array}{ll}
0,& \mbox{ if }  \beta\in\F_{p^n}\setminus\Supp(\Wa f),\\
  \eta_0^{n+s}(-1) (p^*)^{n+s}=p^{n+s},& \mbox{ if }  \beta\in\Supp(\Wa f).
\end{array}\right.
\ee
From the fact that $\#\Supp(\Wa f)=p^{n-s}$ by Lemma \ref{SupportLemma}, the proof is complete.
 The proof of $(iii)$ follows from simple computations (see \cite[Lemma 5]{pott2011association}).
\end{proof}
It is worth noting that the equalities in Lemma \ref{LemmaAsso} should be seen as equalities in the group ring $\C[(\F_{p^n},+)]$. Indeed, the right hand side of the equality in ($ii$) is really $p^n\cdot 0$, where $0$ is the identity of the additive group of  $\F_{p^n}$.
We can derive from Lemma \ref{LemmaAsso} the following reasonable fact.
\begin{lemma}\label{Lemma2Pott}
Let $f\in WRP$, $l$ be given by Proposition \ref{Proposition5} and $D_{f,j}$ be given by (\ref{Subsets}) for every $j\in\F_p$. Then  for $a,b\in\F_p$  
\be\nn
\begin{array}{ll}
p^2D_{f,a}D_{f,b}&= \displaystyle\sum_{k=1}^{p-1}p^{n}\xi_p^{k(a-b)} + (p(\# D_{f,a}+\# D_{f,b})-p^n)\F_{p^n} \\
&+\displaystyle \epsilon \sqrt{p^*}^{n+s}\sum_{\substack{k,t=1\\ k+t\neq 0}}^{p-1} \eta_0^{n+s}(tkv)  \xi_p^{-at-bk} \sL_v,
\end{array} 
\ee
where $v=(k^{1-l}+t^{1-l})^{\frac{1}{1-l}}$.
\end{lemma}
\begin{proof} For $a,b\in\F_p$, by (\ref{DD})   we have
\begin{displaymath}\nn
\begin{array}{ll}
&p^2D_{f,a}D_{f,b} =\displaystyle \sum_{k,t=0}^{p-1}\sL_t\sL_k\xi_p^{-at-bk}\\
=&\sL_0^2
+\displaystyle\sum_{\substack{k,t=1\\ k+t\neq 0}}^{p-1}\sL_t\sL_k\xi_p^{-at-bk}
+\displaystyle\sum_{k=1}^{p-1}\sL_k\sL_{-k}\xi_p^{k(a-b)}
+\displaystyle\sum_{k=1}^{p-1}\sL_k\sL_0\xi_p^{-bk}
+\displaystyle \sum_{t=1}^{p-1}\sL_t\sL_0\xi_p^{-at}\\
=&p^n\F_{p^n}+
\displaystyle \sum_{\substack{k,t, k+t\neq 0\\ k^{1-l}+t^{1-l}=v^{1-l}}}   \epsilon \eta_0^{n+s}(tkv)  \sqrt{p^*}^{n+s} \xi_p^{-at-bk} \sL_v
+\displaystyle\sum_{k=1}^{p-1}p^{n }\xi_p^{k(a-b)}\\
&+(p\# D_{f,b}-p^n)\F_{p^n} + (p\# D_{f,a}-p^n)\F_{p^n},
\end{array} 
\end{displaymath}
which completes the proof.
\end{proof}
We can derive from
Lemma \ref{Lemma2Pott} that 
 for $a,b\in\F_p$,  
\be\label{Combinations}
p^2D_{f,a}D_{f,b}=c+\sum_{j=0}^{p-1}d_jD_{f,j},
\ee
where  $c=p^n(p\delta_{a-b}-1)$ with $\delta$ the Kronecker symbol, and 
\be\nn
d_j=p(\# D_{f,a}+\# D_{f,b})-p^n + \epsilon\sqrt{p^*}^{n+s} \displaystyle \sum_{\substack{k,t=1, k+t\neq 0\\ k^{1-l}+t^{1-l}=v^{1-l}}}^{p-1}    \eta_0^{n+s}(tkv)  \xi_p^{vj-at-bk},
\ee
which implies that $D_{f,a}D_{f,b}$ can be written as a linear combinations of the $D_{f,j}$'s.
%
We now show that $\{0\}, D_{f,0}\setminus \{0\},D_{f,1}, D_{f,2},\ldots,D_{f,p-1}$ span a Schur ring. 
  
\begin{theorem}\label{Schur}
Let $f\in WRP$ and $l$ be given by Proposition \ref{Proposition5}. We define the subsets $A_{f,0}=\{0\}$ and $A_{f,j}=\{ x\in\F_{p^n}^{\star}: f(x)=j-1\}$ for $1\leq j\leq p$.  Then $A_{f,0},A_{f,1},\ldots,A_{f,p}$ span   a Schur ring.
\end{theorem}
\begin{proof}
Notice that $A_{f,1}=D_{f,0}\setminus \{0\}$ and $A_{f,j}=D_{f,j-1}$  for $j\in\{2,\ldots, p\}$.  We then conclude from   (\ref{Combinations}) that $A_{f,a}A_{f,b}$, for $a,b\in\F_p$, can be represented as linear combinations of the $A_{f,j}$'s. Notice that  $A_{f,a}A_{f,b}$ has integer coefficients and the $A_{f,j}$'s are disjoint.  Hence, the proof is complete from the definition. 
 \end{proof}
Theorem \ref{Schur} says that the configuration  $\{\F_{p^n}; G_{0},G_{1},\ldots,G_{p}\}$ is an   association scheme of class $p$ on $\F_{p^n}$, where
$$G_{j}=\{(g,h) : gh^{-1}\in A_{f,j}\}$$
for every $j\in\F_{p}$. The scheme is symmetric since $A_{f,j}^{(-1)}=A_{f,j}$ for each $j\in\F_{p}$.


\subsection{Strongly regular graphs from weakly regular plateaued functions}
In this subsection, we generalize the construction method of strongly regular graphs given by Chee et al.  \cite{chee2011strongly} and by Feng et al. \cite{feng2010partial} from weakly regular bent functions to weakly regular plateaued functions over finite fields of odd characteristic. We  show that for an $s$-plateaued $f\in WRP$ 
  the subsets
\be\label{DifferenceSets}
\begin{array}{ll}
D_{f}&=\{ x\in\F_{p^n}^{\star}: f(x)=0\},\\
D_{f,sq}&=\{ x\in\F_{p^n}^{\star}: f(x)\in SQ\},\\
D_{f,sq,0}&=\{ x\in\F_{p^n}^{\star}: f(x)\in SQ\cup \{0\}\},\\
D_{f,nsq}&=\{ x\in\F_{p^n}^{\star}: f(x)\in NSQ\}
\end{array}
\ee
are partial difference sets in $(\F_{p^n},+)$ when $n+s$ is an even integer. Actually, we have
\be\nn
\begin{array}{ll}
D_{f}=D_{f,0}\setminus \{0\},\;\;
 D_{f,sq}=  \displaystyle\bigcup_{j\in SQ}D_{f,j}, \;\;
D_{f,sq,0}= \displaystyle\bigcup_{j\in SQ}D_{f,j}\cup D_{f} \mbox{ and }
D_{f,nsq}=\displaystyle \bigcup_{j\in NSQ}D_{f,j}
\end{array}
\ee
where  $D_{f,j}$  is defined in (\ref{Subsets})  for every $j\in\F_{p}$.
Clearly, we have $-D_{f}=D_{f}$, $- D_{f,sq}= D_{f,sq}$, $-D_{f,sq,0}=D_{f,sq,0}$ and $-D_{f,nsq}=D_{f,nsq}$ since $f(x)=f(-x)$ for every $x\in\F_{p^n}$.
\begin{lemma}\cite{chee2011strongly}\label{Lemma3Pott} Let $n+s$ be an even integer,  $f\in WRP$ and $l$ be given by Proposition \ref{Proposition5}. Then
we have
\be\nn 
\displaystyle \sum_{\substack{k,t=1\\
k+t\neq 0}}^{p-1} \sL_v
=(p-2) \left(pD_{f,0}-\F_{p^n}\right),
\ee
where $v= (k^{1-l}+t^{1-l})^\frac{1}{1-l}$.
\end{lemma}
We now  establish the main results of this subsection, which require to compute the squares $D_{f}^2$, $D_{f,sq}^2$, $D_{f,sq,0}^2$ and $D_{f,nsq}^2$ in $\C[(\F_{p^n},+)]$.
\begin{theorem}\label{Theorem0}
Let $n+s$ be an even integer, $f\in WRP$ and $D_{f}$ defined in (\ref{DifferenceSets}). Then
$D_{f}$ is a $(v,d,\lambda_1,\lambda_2)$-PDS in $(\F_{p^n},+)$, where
\be\nn
\begin{array}{ll}
v&=p^{n},\\
d&=
p^{n-1}+\epsilon \eta_0 (-1)(p-1)\sqrt{p^*}^{n+s-2}-1,\\
\lambda_1&=
p^{n-2}+\epsilon \eta_0 (-1)(p-1)\sqrt{p^*}^{n+s-2}-2,\\
\lambda_2&=
p^{n-2}+\epsilon \eta_0 (-1) \sqrt{p^*}^{n+s-2}.
\end{array}
\ee

\end{theorem}

\begin{proof}
By Lemma  \ref{Lemma2Pott}, for $a=b=0$, we clearly  have
\be\nn
\begin{array}{ll}
p^2D_{f,0}D_{f,0}&=(p-1)p^n+
(2p\# D_{f,0}-p^n)\F_{p^n}+  \epsilon \sqrt{p^*}^{n+s}
\displaystyle \sum_{\substack{k,t, k+t\neq 0\\ k^{1-l}+t^{1-l}=v^{1-l}}} \sL_v\\
&=(p-1)p^n+\epsilon(p-2)p  \sqrt{p^*}^{n+s}D_{f,0} + 
(2p\# D_{f,0}-p^n-\epsilon \sqrt{p^*}^{n+s}(p-2))\F_{p^n},
\end{array}
\ee
where we used    Lemma \ref{Lemma3Pott} in the last equality. Recall that  $\#D_{f,0}=p^{n-1}+\epsilon \eta_0 (-1)(p-1)\sqrt{p^*}^{n+s-2}$ by Lemma \ref{Lemma7} and  $D_{f}=D_{f,0}\setminus  \{0\}$.  Then we get
\be\nn
\begin{array}{ll}
p^2(D_{f}+0)^2&
=(p-1)p^n+\epsilon(p-2)p  \sqrt{p^*}^{n+s}(D_{f}+0) \\
&+(2p^{n}+\epsilon 2\eta_0 (-1)p(p-1)\sqrt{p^*}^{n+s-2}-p^n-\epsilon \sqrt{p^*}^{n+s}(p-2))\F_{p^n}.
\end{array}
\ee
With the simple computation, we obtain the  following equation 
 \be\nn
\begin{array}{ll}
D_{f}^2&
=\left(p^{n-2}(p-1)+\epsilon \eta^{\frac{n+s}{2}}_0(-1)  (p-2)p^{\frac{n+s-2}{2}}-1\right)\\
&+\left(\epsilon\eta^{\frac{n+s}{2}}_0(-1)(p-2) p^{\frac{n+s-2}{2}}-2\right)D_{f}
+\left(p^{n-2}+\epsilon\eta^{\frac{n+s}{2}}_0(-1)p^{\frac{n+s-2}{2}}\right)\F_{p^n},
\end{array}
\ee
from which one can deduce  the desired parameters by  (\ref{Condition}).
\end{proof}
  
The following lemma follows directly from  \cite[Lemma 2 and the proof of Theorem 2]{chee2011strongly}. 
 \begin{lemma}\label{LemmaPott} 
Let $n+s$ be an even integer,   $f\in WRP$ and $l$ be given by Proposition \ref{Proposition5}. Then we have
\be\nn
\displaystyle\sum_{a,b\in SQ}\sum_{
\substack{k,t=1 \\
k+t\neq 0 }}^{p-1}  \xi_p^{-at-bk} \sL_v=pD_{f,sq}-\frac{(p-1)}{2}\F_{p^n},
\ee
where $v=(k^{1-l}+t^{1-l})^{\frac{1}{1-l}}$. 
\end{lemma}

\begin{theorem}\label{TheoremSQ}
Let $n+s$ be an even integer, $f\in WRP$ and $D_{f,sq}$ defined by (\ref{DifferenceSets}). Then
$D_{f,sq}$ is a $(v,d,\lambda_1,\lambda_2)$-PDS in $(\F_{p^n},+)$, where
\be\nn
\begin{array}{ll}
v&=p^{n},\\
d&=\frac{(p-1)}{2}\left(p^{n-1}  -\epsilon\eta_0(-1)\sqrt{p^*}^{n+s-2}\right),\\
\lambda_1&=\frac{1}{4}p^{n-2}(p-1)^2  -\epsilon\eta_0(-1)\frac{1}{2} (p-3)  \sqrt{p^*}^{n+s-2},\\
\lambda_2&=\frac{1}{4}p^{n-2}(p-1)^2  - \epsilon\eta_0(-1)\frac{1}{2} (p-1)  \sqrt{p^*}^{n+s-2}.
\end{array}
\ee

\end{theorem}
\begin{proof}
 By Lemma \ref{Lemma2Pott}, by noting that $n+s$ is even,
 we get that 
\begin{displaymath}\label{mm}
\begin{array}{ll}
p^2D_{f,sq}^2 &=  p^2\displaystyle\sum_{a,b\in SQ} D_{f,a}D_{f,b}\\
&=\displaystyle \sum_{a,b\in SQ}  \left(
\displaystyle\sum_{k=1}^{p-1} p^{n } \xi_p^{k(a-b)}+ (2p\# D_{f,1}-p^n)\F_{p^n}+  \epsilon   \sqrt{p^*}^{n+s} \displaystyle\sum_{\substack{k,t=1 \\k+t\neq 0
 }}^{p-1}   \xi_p^{-at-bk} \sL_v\right),
\end{array} 
\end{displaymath}
where $v=(k^{1-l}+t^{1-l})^{\frac{1}{1-l}}$.
One can easily observe that 
\be\label{AA}
\begin{array}{ll}
\displaystyle \sum_{a,b\in SQ}\sum_{k=1}^{p-1} p^{n } \xi_p^{k(a-b)}&=\frac{p^{n }(p-1)(p+1)}{4},\\
\displaystyle \sum_{a,b\in SQ} (2p\# D_{f,1}-p^n)\F_{p^n}&=\left(\frac{p^n(p-1)^2}{4}- \epsilon \frac{(p-1)^2}{2} \sqrt{p^*}^{n+s} \right)\F_{p^n},
\end{array}
\ee
where $\# D_{f,1}=p^{n-1}-\epsilon \eta_0 (-1)\sqrt{p^*}^{n+s-2}$ by Lemma \ref{Lemma7}.
Hence, combining Lemma \ref{LemmaPott} and the above results in (\ref{AA}), we get   
\begin{displaymath}\nn
\begin{array}{ll}
 p^2D_{f,sq}^2& = \frac{p^{n}(p-1)(p+1)}{4}+\left(\frac{p^n(p-1)^2}{4}- \epsilon \frac{(p-1)^2}{2} \sqrt{p^*}^{n+s} \right)\F_{p^n}+ \epsilon \sqrt{p^*}^{n+s}\left(pD_{f,sq}-\frac{p-1}{2}\F_{p^n}\right)\\
&=C_1 + C_2 D_{f,sq} + C_3 \F_{p^n},
\end{array} 
\end{displaymath}
where 
$
C_1= \frac{1}{4}p^{n}(p-1)(p+1), 
C_2=\epsilon p\sqrt{p^*}^{n+s}$ and $
C_3= \frac{1}{4}p^n(p-1)^2  - \epsilon\frac{1}{2}p(p-1)  \sqrt{p^*}^{n+s}$
from which one can deduce  the desired parameters by  (\ref{Condition}).
\end{proof}

\begin{remark}\label{TheoremNSQ}  With the same proof of  Theorem \ref{TheoremSQ}, one can observe that 
$D_{f,nsq}$  defined in (\ref{DifferenceSets}) is a $(v,d,\lambda_1,\lambda_2)$-PDS in $(\F_{p^n},+)$ with the same parameters of Theorem \ref{TheoremSQ}.
\end{remark}

The following theorem follows directly from Theorem \ref{TheoremSQ} with a  slight modification.

\begin{theorem}\label{TheoremSQ0}
Let $n+s$ be an even integer, $f\in WRP$ and $D_{f,sq,0}$  defined in (\ref{DifferenceSets}). Then
$D_{f,sq,0}$ is a $(v,d,\lambda_1,\lambda_2)$-PDS in $(\F_{p^n},+)$, where  
\be\nn
\begin{array}{ll}
v&=p^{n},\\
d&= \frac{(p+1)}{2} p^{n-1} +\epsilon\eta_0(-1)\frac{(p-1)}{2} \sqrt{p^*}^{n+s-2}-1,\\
\lambda_1&=\frac{1}{4}p^{n-2}(p-1)^2 -2 + \epsilon\eta_0(-1)\frac{(p-1) }{2} \sqrt{p^*}^{n+s-2} ,\\
\lambda_2&=\frac{1}{4}p^{n-2}(p-1)^2 + \epsilon\eta_0(-1)\frac{(p-1) }{2} \sqrt{p^*}^{n+s-2}.
\end{array}
\ee
\end{theorem}
\begin{proof}
We   have clearly that $D_{f,sq,0}=D_{f,sq}\cup  D_{f} $ from  (\ref{DifferenceSets}). Then 
$$D_{f,sq,0}^2=(D_{f} + D_{f,sq})(D_{f} + D_{f,sq}).$$ Hence, the desired parameters follow from Theorems \ref{Theorem0} and \ref{TheoremSQ} with the same group ring computations.
\end{proof}

%

\section{Conclusion}
In this paper, inspired by the works of \cite{chee2011strongly,pott2011association}, to construct association schemes and strongly regular graphs we push further the use of weakly regular plateaued functions over finite fields of odd characteristic  introduced by Mesnager et al.    \cite{mesnager2017WCC}.  By generalizing the construction methods given in \cite{chee2011strongly,pott2011association}, we obtain  a family of  $p$-class association schemes and  strongly regular graphs with three types of parameters  from weakly regular $p$-ary plateaued functions, where $p$ is an odd prime.
The paper provides the first constructions of  association schemes and strongly regular graphs from weakly regular plateaued functions over finite fields of odd characteristic.


\bibliography{myBiblio}
\bibliographystyle{abbrv}
\bibliographystyle{iamBiblioStyle}

\end{document}